\documentclass[twocolumn]{IEEEtran}
\usepackage[final]{graphicx}
\usepackage{amsthm}
\usepackage{amsmath,epsfig,amssymb,verbatim,amsopn,cite,multirow}
\usepackage{balance}
\usepackage{multirow}
\usepackage{footnote}
\usepackage{stfloats}
\usepackage{algorithm}
\usepackage{algorithmic}
\usepackage[usenames,dvipsnames]{color}
\usepackage[all]{xy}
\usepackage{url}
\usepackage{subcaption}
\usepackage{graphicx}
\usepackage{array}
\usepackage{threeparttable}

\usepackage[nodisplayskipstretch]{setspace}
  {\proof}{\proofend}
\newtheorem{proposition}{Proposition}

\newcommand{\qg}{{\bf g}}
\newcommand{\qh}{{\bf h}}

\newcommand{\qw}{{\bf w}}
\newcommand{\qx}{{\bf x}}
\newcommand{\qy}{{\bf y}}

\newcommand{\qF}{{\bf F}}

\newcommand{\dlj}{\mathtt{J}}
\newcommand{\dld}{\mathtt{D}}
\newcommand{\ul}{\mathtt{O}}
\newcommand{\lu}{\mathtt{U}}
\newcommand{\gmlj}{\qg_{mu}^{\dlj}}
\newcommand{\gmplj}{\qg_{m'u}^{\dlj}}
\newcommand{\gmkd}{\qg_{mk}^{\dld}}
\newcommand{\gmpkd}{\qg_{m'k}^{\dld}}
\newcommand{\guk}{g_{uk}}
\newcommand{\gmuo}{\qg_{mu}^{\ul}}
\newcommand{\gmpudu}{\qg_{m'u}^{\ul}}
\newcommand{\gmpupu}{\qg_{m'u'}^{\ul}}
\newcommand{\gmupu}{\qg_{mu'}^{\ul}}
\newcommand{\gulu}{g_{u}^{\lu}}
\newcommand{\guplu}{g_{u'}^{\lu}}
\newcommand{\hgmkd}{\hat{\qg}_{mk}^{\dld}}
\newcommand{\hgmKd}{\hat{\qg}_{mK}^{\dld}}
\newcommand{\hgmod}{\hat{\qg}_{m1}^{\dld}}

\newcommand{\hgmlj}{\hat{\qg}_{mu}^{\dlj}}
\newcommand{\hgmLj}{\hat{\qg}_{mU}^{\dlj}}

\newcommand{\hgmoj}{\hat{\qg}_{m1}^{\dlj}}

\newcommand{\hgmuo}{\hat{\qg}_{mu}^{\ul}}
\newcommand{\hgmoo}{\hat{\qg}_{m1}^{\ul}}
\newcommand{\hgmUo}{\hat{\qg}_{mU}^{\ul}}

\newcommand{\betamuo}{\beta_{mu}^{\ul}}

\newcommand{\betamupo}{\beta_{mu'}^{\ul}}
\newcommand{\betamlj}{\beta_{mu}^{\dlj}}

\newcommand{\betamkd}{\beta_{mk}^{\dld}}

\newcommand{\betaulu}{\beta_{u}^{\lu}}
\newcommand{\betauplu}{\beta_{u'}^{\lu}}
\newcommand{\betami}{\beta_{mi}}

\newcommand{\hmuo}{{\qh}_{mu}^{\ul}}

\newcommand{\hmlj}{{\qh}_{mu}^{\dlj}}

\newcommand{\hmkd}{{\qh}_{mk}^{\dld}}

\newcommand{\hulu}{h_{u}^{\lu}}

\newcommand{\gamamkd}{\gamma_{mk}^{\dld}}
\newcommand{\gamampkd}{\gamma_{m'k}^{\dld}}
\newcommand{\gamamkpd}{\gamma_{mk'}^{\dld}}
\newcommand{\gamamlj}{\gamma_{mu}^{\dlj}}

\newcommand{\gamamlpj}{\gamma_{mu'}^{\dlj}}
\newcommand{\gamamuo}{\gamma_{mu}^{\ul}}

\newcommand{\gamampuo}{\gamma_{m'u}^{\ul}}
\newcommand{\gamampupo}{\gamma_{m'u'}^{\ul}}

\newcommand{\UR}{\mathtt{U}}

\newcommand{\MR}{\mathtt{MR}}
\newcommand{\PZF}{\mathtt{PZF}}

\newcommand{\bmkpzf}{\boldsymbol{b}_{mk'}^\PZF}
\newcommand{\bikpzf}{\boldsymbol{b}_{ik'}^\PZF}
\newcommand{\bikzf}{\boldsymbol{b}_{ik}^\PZF}
\newcommand{\bmkzf}{\boldsymbol{b}_{mk}^\PZF}
\newcommand{\bmkmr}{\boldsymbol{b}_{mk}^\MR}
\newcommand{\bikmr}{\boldsymbol{b}_{ik}^\MR}
\newcommand{\bipkmr}{\boldsymbol{b}_{i'k}^\MR}
\newcommand{\bmpkmr}{\boldsymbol{b}_{m'k}^\MR}
\newcommand{\bmkpmr}{\boldsymbol{b}_{mk'}^\MR}
\newcommand{\bmpkpmr}{\boldsymbol{b}_{m'k'}^\MR}

\newcommand{\bmlpzf}{\boldsymbol{b}_{mu'}^\PZF}
\newcommand{\bilpzf}{\boldsymbol{b}_{iu'}^\PZF}
\newcommand{\bilzf}{\boldsymbol{b}_{iu}^\PZF}
\newcommand{\bmlzf}{\boldsymbol{b}_{mu}^\PZF}

\newcommand{\bmlmr}{\boldsymbol{b}_{mu}^\MR}
\newcommand{\bilmr}{\boldsymbol{b}_{iu}^\MR}
\newcommand{\biplmr}{\boldsymbol{b}_{i'u}^\MR}
\newcommand{\bmplmr}{\boldsymbol{b}_{m'u}^\MR}
\newcommand{\bmplpmr}{\boldsymbol{b}_{m'u'}^\MR}
\newcommand{\vmuo}{\mathbf{v}_{mu}^\ul}
\newcommand{\vmumr}{\mathbf{v}_{mu}^\MR}

\newcommand{\vmpumr}{\mathbf{v}_{m'u}^\MR}

\newcommand{\vmuzf}{\mathbf{v}_{mu}^\PZF}

\newcommand{\Sm}{\mathcal{S}_m}
\newcommand{\Si}{\mathcal{S}_i}
\newcommand{\Wm}{\mathcal{W}_m}
\newcommand{\Wi}{\mathcal{W}_i}

\DeclareMathOperator*{\argmax}{arg\,max}

\usepackage[top=0.75 in,bottom=1 in, left=0.625 in,  right=0.625 in]{geometry}

\author{\IEEEauthorblockN{Mustafa S. Abbas\IEEEauthorrefmark{1}, 
Zahra Mobini\IEEEauthorrefmark{2},
Hamid Reza Hashempour\IEEEauthorrefmark{1},
Hien Quoc Ngo\IEEEauthorrefmark{1}, and Michail Matthaiou\IEEEauthorrefmark{1}}

\IEEEauthorblockA{\IEEEauthorrefmark{1}Centre for Wireless Innovation (CWI), Queen’s University Belfast, U.K. \\
\IEEEauthorrefmark{2}Department of Electrical and Electronic Engineering, The University of Manchester, U.K. \\
Email: \text{\{malkhadhrawee01, h.hashempoor, hien.ngo, m.matthaiou\}@qub.ac.uk,}
\text{zahra.mobini@manchester.ac.uk}
}
}

\begin{document}



\title{Cell-Free Massive MIMO for Joint Communication and Proactive Monitoring}

\maketitle

\begin{abstract}
This paper introduces a novel joint communication and proactive monitoring (JCAM) system that simultaneously monitors multiple untrusted links and serves multiple legitimate users. The system leverages a cell-free massive multiple-input multiple-output (CF-mMIMO) architecture, where one subset of access points (APs) is dedicated to receiving signals from untrusted links, while another subset transmits data to legitimate  users and jamming signals into the untrusted links. This dual functionality not only ensures reliable communication for legitimate users but also degrades the performance of untrusted links, thereby enhancing monitoring effectiveness. Closed-form expressions for the spectral efficiency (SE) of legitimate users and the monitoring success probability (MSP) are derived under partial zero-forcing (PZF) precoding/combining schemes with imperfect channel state information. Leveraging these expressions, we develop a simple yet effective AP mode assignment strategy that determines which APs perform downlink transmission and jamming, and which APs are dedicated to receiving signals from untrusted links. The objective is to maximize the MSP while satisfying predefined quality-of-service (QoS) requirements for all legitimate users. Numerical results show that the proposed mode assignment strategy significantly outperforms the benchmark, achieving up to a $32\%$ improvement in monitoring performance, while maintaining low computational complexity.  Moreover, our proposed JCAM framework provides nearly a six-fold improvement in the minimum MSP over the co-located massive MIMO baseline.

\let\thefootnote\relax\footnotetext{ 
This work was supported by the U.K. Engineering and Physical Sciences Research Council (EPSRC) grant (EP/X04047X/2) for TITAN Telecoms Hub.
The work of M. S. Abbas, Z. Mobini, H. R. Hashempour and H. Q. Ngo was supported by the U.K. Research and Innovation Future Leaders Fellowships under Grant MR/X010635/1, and a research grant from the Department for the Economy Northern Ireland under the US-Ireland R\&D Partnership Programme. The work of M. Matthaiou has received funding from the European Research Council (ERC) under the European Union’s Horizon 2020 research and innovation programme (grant agreement No. 101001331).}

\end{abstract}
 \begin{IEEEkeywords}
 Cell-free massive multiple-input multiple-output (CF-mMIMO), partial zero-forcing, proactive monitoring system. 
 \end{IEEEkeywords}

\vspace{-1em}
\section{Introduction}
The evolution of wireless communication technologies has advanced significantly across multiple generations, accompanied by a growing demand for enhanced wireless security. However, infrastructure-free or user-controlled networks, such as device-to-device (D2D) and mobile ad hoc networks, pose amplified public safety risks due to their decentralized nature. These networks can be exploited by malicious users to conduct illicit activities, cybercrimes, or other threats. In response, numerous studies have investigated technologies aimed at mitigating unauthorized or untrusted communications within wireless networks~\cite{Mohammadi:Proc:2024}.

One of the most promising techniques grounded in physical-layer security (PLS) is proactive monitoring \cite{Xu:TWC:2017}, which enhances network security by jamming untrusted receivers in the downlink and monitoring untrusted transmitters in the uplink, thereby increasing the MSP. A MIMO proactive monitoring system was considered in \cite{Zhong:TWC:2017,yao2024proactive}, where a legitimate monitor eavesdrops on a suspicious transmitter–receiver pair. The study aimed to maximize the non-outage probability by jointly designing the jamming power and the transmit/receive beamformers at the monitoring node. In \cite{Feizi:TCOMM:2020}, a MIMO proactive monitoring system was investigated in which the transmit and receive beamformers at a legitimate full-duplex (FD) monitor were jointly optimized to maximize the eavesdropping non-outage probability.

A key challenge in proactive monitoring is that untrusted links are randomly located across wide areas, making it difficult to develop a system that can effectively monitor such links irrespective of location. To address this, CF-mMIMO-based proactive monitoring systems were proposed in \cite{Zahra:IOT:2024, da2025proactively}. In particular\cite{Zahra:IOT:2024}, a CF-mMIMO surveillance framework was proposed using maximum ratio (MR) and PZF combining schemes, jointly optimizing power control and weighting coefficients at the monitoring APs to improve the MSP. In \cite{da2025proactively}, an effective channel state information (CSI) acquisition scheme for CF-mMIMO monitoring was proposed. However, these works focused only on untrusted links, while, in practice, multiple legitimate users must also be served.

Motivated by this gap, in this paper we propose a novel JCAM system, which leverages CF-mMIMO to simultaneously serve multiple users and monitor multiple untrusted links. The main contributions of this paper are as follows.
\vspace{-1em}
\begin {itemize}
\item  We derive closed-form expressions for the signal-to-interference-plus-noise ratio  (SINR) in a CF-mMIMO-based JCAM system with multiple downlink legitimate users, untrusted transmitters, and untrusted receivers. The analysis employs PZF precoding for downlink transmission and PZF combining for monitoring, utilizing the use-and-then-forget bounding technique. Furthermore, we formulate the MSP to quantify the likelihood of successful monitoring in the proposed JCAM system.

\item We propose a simple and effective AP mode selection algorithm to enhance monitoring performance while ensuring a predefined QoS for each legitimate user. The scheme focuses on efficiently determining which APs are assigned to overhear signals from untrusted links and which APs are designated for downlink transmission.

\item Our numerical results show that the proposed CF-mMIMO JCAM system significantly improves monitoring performance compared to mMIMO-based systems relying on FD operation, where all APs are consolidated into an antenna array performing observation, communication, and jamming simultaneously. Furthermore, the results demonstrate that the proposed mode assignment algorithm substantially enhances the minimum MSP in the considered CF-mMIMO-based JCAM system compared to a randomly assigned baseline. 
\end {itemize}

\textit{Notation:} We use bold lowercase letters (uppercase) to denote vectors (matrices);  $\mathbf{I}_N$ denotes the $N\times N$ identity matrix; $(\cdot)^{-1}$ denotes the matrix inverse; the superscript $(\cdot)^H$ stands for the Hermitian transpose; $\mathcal{CN}(0,\sigma^2)$ denotes a complex circularly symmetric Gaussian random variable with variance $\sigma^2$. Finally, $\mathbb{E}\{\cdot\}$ denotes the statistical expectation.
\vspace{-0.5em}
\section{System Model}\label{sysmod}
We consider a CF-mMIMO-based JCAM system that performs joint proactive monitoring (i.e., monitoring $U$  untrusted  links) and  downlink payload data transmission to $K$ legitimate users. The system consists of $M$ half-duplex APs, each equipped with $N$ antennas, and operates in two modes: downlink mode and monitoring mode. In downlink mode, a subset of APs simultaneously serves $K$ downlink legitimate users and transmits jamming signals to $U$ untrusted receivers. In monitoring mode, the remaining APs observe $U$ untrusted transmitters, where the $u$-th untrusted transmitter  intends to transmit a signal to the $u$-th untrusted receiver. The sets of APs, downlink users, untrusted pairs are denoted by $\mathcal{M}$, $\mathcal{K}$, $\mathcal{U}$, respectively. 

We define a binary variable $a_{m}$ to indicate the operational mode of the $m$-th AP, where each mode corresponds to either downlink transmission or monitoring, and is given by
\begin{align}
a_{m} \triangleq
\begin{cases}
  1, & \text{if AP $m$ operates in the downlink mode,}\\
  0, & \mbox{if AP $m$ operates in the monitoring mode}.
\end{cases} 
\end{align}
\vspace{-0.9em}
\subsection{Channel Model and Channel Estimation}
The channel matrix between the $m$-th AP operating in downlink mode and the $i$-th AP operating in monitoring mode, $\forall m,i\in\mathcal{M}$, is denoted by $\qF_{mi}\in \mathbb{C}^{N \times N}$. The elements of this matrix are distributed as $\mathcal{CN}(0,\betami)$ for $i\neq m$. The channel vector between the $m$-th AP operating in downlink mode, where $m \in \mathcal{M}$, and the $k$-th downlink communication user, where $k \in \mathcal{K}$, is $ \gmkd = \sqrt{\betamkd} \hmkd\in\mathbb{C}^{N \times 1}$. Moreover, the jamming channel vector between the $m$-th AP operating in downlink mode and the $u$-th untrusted receiver, where $u \in \mathcal{U}$, is $\gmlj = \sqrt{\betamlj} \hmlj\in\mathbb{C}^{N \times 1}$.

The monitoring channel between the $m$th AP operating in monitoring mode, where $m \in \mathcal{M}$, and the $u$-th untrusted transmitter, where $u \in \mathcal{U}$, is $\gmuo = \sqrt{\betamuo}\hmuo\in\mathbb{C}^{N \times 1}$. Moreover, the channel between the $u$-th untrusted transmitter  and the $u$-th untrusted receiver  is $\gulu = \sqrt{\betaulu} \hulu $. In this context, $\betami$, $\betamkd$, $\betamlj$, $\betamuo$, and $\betaulu$ denote the large-scale fading coefficients. Additionally, the small-scale fading vectors are modeled as follows: $\hmkd \sim \mathcal{CN} (\mathbf 0,\mathbf I_{N})$, $\hmlj \sim \mathcal{CN} (\mathbf 0,\mathbf I_{N})$, $\hmuo \sim \mathcal{CN} (\mathbf 0,\mathbf I_{N})$, and $\hulu \sim \mathcal{CN} (0,1)$.
Following the minimum mean square error (MMSE) estimation method in \cite{Hien:JWCOM:2017,Zahra:IOT:2024}, the estimated channels of $\gmkd$, $\gmlj$, and $\gmuo$ are modeled respectively as 
$\hgmkd \sim \mathcal{CN}(\mathbf{0},\gamamkd \mathbf{I}_N)$, $\hgmlj \sim \mathcal{CN}(\mathbf{0},\gamamlj \mathbf{I}_N)$, $\hgmuo \sim \mathcal{CN}(\mathbf{0},\gamamuo \mathbf{I}_N)$, where $\gamamkd$, $\gamamlj$, and $\gamamuo$ are given by $
 \gamamkd = \frac{\tau\rho_{\text{u}} (\betamkd)^2}{\tau\rho_{\text{u}}\betamkd+1}$, $
\gamamlj = \frac{\tau\rho_{\text{u}}(\betamlj)^2}{\tau\rho_{\text{u}}\betamlj+1}$, and $
\gamamuo = \frac{\tau\rho_{\text{u}}(\betamuo)^2}{\tau\rho_{\text{u}}\betamuo+1}$, where $\rho_{\text{u}}$ is  the   normalized transmit power of each pilot symbol  and $\tau$ is the pilot length, which satisfies the condition $K+U\leq\tau \leq T$, where $T$ is the coherence interval. Note that the untrusted links require a training phase to acquire their channels for their own transmissions. During this training phase, the APs can estimate the channels to the untrusted nodes by eavesdropping on the pilot signals transmitted by the untrusted links \cite{Zahra:IOT:2024}.
\vspace{-0.8em}
\subsection{Downlink Payload Data and Jamming Signal Transmission}
 For downlink transmission, we employ the PZF scheme relying only on local channel knowledge  \cite{Emil:TWC:2020,Zhang:TCOMM:2021}. PZF serves as a general framework encompassing both MR and ZF, enabling the system to transition dynamically between these extremes based on user density and channel conditions. 
In particular, the $m$-th AP operating in downlink mode classifies the downlink communication users into two groups based on their large-scale fading coefficients: 1) $\Sm^\dld$ comprising strong downlink users, and 2) $\Wm^\dld$ comprising weak downlink users, where $\Sm^\dld \bigcap \Wm^\dld = \varnothing$. Similarly, each $m$-th AP classifies the untrusted receivers into: 1)  $\Sm^\dlj$,  comprising strong untrusted receivers, and 2)  $\Wm^\dlj$,  comprising weak untrusted receivers, for the purpose of directing jamming signals to reduce the SINR at these untrusted links, where $\Sm^\dlj \bigcap \Wm^\dlj={\varnothing}$. For the sets $\Sm^\dld$ and $\Sm^\dlj$, the $m$-th AP applies PZF precoding, whereas for $\Wm^\dld$ and $\Wm^\dlj$, it employs MR precoding.
Let $s_{k}^{\dld}$ and $s_{u}^{\dlj}$ denote the symbols allocated to the $k$-th downlink communication user and the $u$-th untrusted receiver, respectively, with $\mathbb{E}\big\{\big|s_{k}^{\dld}\big|^2\big\}=\mathbb{E}\big\{\big|s_{u}^{\dlj}\big|^2\big\}=1$ and $\mathbb{E}\left\{s_{k}^{\dld}\right\}=\mathbb{E}\left\{s_{u}^{\dlj}\right\}=0$. Then, the transmitted signal from the $m$-th AP is
\begin{align}\label{xmpzf}
    \qx_m^\PZF \!&=\! a_{m} \sqrt{\eta\rho_{\text{d}}} \Big(\!\sum\nolimits_{k'\in \Sm^\dld} \!\! \bmkpzf s_{k'}^{\dld}\!+\!\!\sum\nolimits_{k\in \Wm^\dld}\!\!  \bmkmr s_{k}^{\dld}\nonumber\\
    &\hspace{3em}+ \sum\nolimits_{u'\in\Sm^\dlj}\!\! \bmlpzf s_{u'}^{\dlj}\!+\!\sum\nolimits_{u\in\Wm^\dlj}\!\! \bmlmr s_{u}^{\dlj}\Big),
\end{align}
where $\rho_{\text{d}}$ is the normalized transmit power at each AP,   while $\eta=\frac{1}{K+U}$  denotes the normalization coefficient that guarantees $\mathbb{E}\{\|\qx_{m}^\PZF\|^2\}  = \rho_{\text{d}}$. 
Moreover, $\bmkmr$ and $\bmkpzf$ represent the   precoding vectors for the downlink users, and are given by
$ \bmkmr=\frac{\mathbf{\hat{G}}_m^\dld \mathbf{e}_{k}^\dld}{\sqrt{\mathbb{E}\big\{\big\|\mathbf{\hat{G}}_m^\dld \mathbf{e}_{k}^\dld\big\|^2\big\}}},$
and
$\bmkpzf=\frac{\boldsymbol{\theta}_{mk'}^\dld }{\sqrt{\mathbb{E} \big\{ \big\|\boldsymbol{\theta}_{mk'}^\dld \big\|^2\big\}}}.$
Here, $\boldsymbol{\theta}_{mk'}^\dld = \mathbf{\hat{G}}_m^\dld  \mathbf{\Upsilon}_{\Sm^\dld}^\dld((\mathbf{\Upsilon}_{\Sm^\dld}^\dld)^H (\mathbf{\hat{G}}_m^\dld)^H\mathbf{\hat{G}}_m^\dld\mathbf{\Upsilon}_{\Sm^\dld}^\dld)^{-1} \boldsymbol{\varepsilon}_{k'}^\dld$ and $\mathbf{\hat{G}}_m^\dld=\left[ \hgmod,\dots, \hgmKd\right]_{\left(N \times K\right)}$. The vector $\mathbf{e}_{k}^\dld$ denotes the $k$-th column of the identity matrix $\mathbf{I}_{K}=\left[\mathbf{e}_{1}^\dld, \mathbf{e}_{2}^\dld, \dots, \mathbf{e}_{K}^\dld \right]_{\left(K \times K\right)}$. The matrix $\mathbf{\Upsilon}_{\Sm^\dld}^\dld=\big[ \mathbf{e}_{1'}^\dld, \mathbf{e}_{2'}^\dld, \dots, \mathbf{e}_{|\Sm^\dld|}^\dld \big]_{\left(K \times |\Sm^\dld|\right)}$ is constructed by selecting the columns of $\mathbf{I}_{K}$ corresponding to the users in $\Sm^\dld$. The vector $\boldsymbol{\varepsilon}_{k'}^\dld$ denotes the $k'$-th column of $\mathbf{\Upsilon}_{\Sm^\dld}^\dld$. The normalization terms are   given by       $\mathbb{E}\Big\{\|\mathbf{\hat{G}}_m^\dld \mathbf{e}_{k}^\dld\|^2\Big\}= N \gamamkd$ and $\mathbb{E} \left\{\left\| \boldsymbol{\theta}_{mk'}^\dld \right\|^2\right\} = \frac{1}{(N-\left|\Sm^\dld \right|)\gamamkpd}$, respectively.
Furthermore, the precoding vectors $\bmlpzf$ and $\bmlmr$ associated with the untrusted receiver, are obtained respectively by 
$\bmlmr=\frac{\mathbf{\hat{G}}_m^\dlj \mathbf{e}_{u}^\dlj}{\sqrt{\mathbb{E}\left\{\left\|\mathbf{\hat{G}}_m^\dlj \mathbf{e}_{u}^\dlj \right\|^2\right\}}}$
and
$\bmlpzf=\frac{\boldsymbol{\theta}_{mu'}^\dlj}{\sqrt{\mathbb{E}\left\{\left\|\boldsymbol{\theta}_{mu'}^\dlj \right\|^2\right\}}}$, where $\boldsymbol{\theta}_{mu'}^\dlj=\mathbf{\hat{G}}_m^\dlj \mathbf{\Upsilon}_{\Sm^\dlj}^\dlj ((\mathbf{\Upsilon}_{\Sm^\dlj}^\dlj)^H(\mathbf{\hat{G}}_m^\dlj)^H\mathbf{\hat{G}}_m^\dlj \mathbf{\Upsilon}_{\Sm^\dlj}^\dlj)^{-1} \boldsymbol{\varepsilon}_{u'}^\dlj$ and $\mathbf{\hat{G}}_m^\dlj=\left[ \hgmoj,\dots, \hgmLj\right]_{\left(N \times {U} \right)}$. The vector $\mathbf{e}_{u}^\dlj$ denotes the $u$-th column of the identity matrix $\mathbf{I}_{U}=\left[\mathbf{e}_{1}^\dlj, \mathbf{e}_{2}^\dlj, \dots, \mathbf{e}_{U}^\dlj \right]_{\left(U \times U\right)}$. The matrix $\mathbf{\Upsilon}_{\Sm^\dlj}^\dlj=\big[ \mathbf{e}_{1'}^\dlj, \mathbf{e}_{2'}^\dlj, \dots, \mathbf{e}_{|\Sm^\dlj|}^\dlj \big]_{\left(U \times |\Sm^\dlj|\right)}$ is constructed by selecting the columns of $\mathbf{I}_{U}$ corresponding to the untrusted received units in $\Sm^\dlj$. The vector $\boldsymbol{\varepsilon}_{u'}^\dlj$ denotes the $u'$-th column of the $\mathbf{\Upsilon}_{\Sm}^\dlj$. The normalization terms in $\bmlmr$ and $\bmlpzf$ are given by $\mathbb{E} \big\{\big\| \boldsymbol{\theta}_{mu'}^\dlj \big\|^2\big\} = \frac{1}{(N-\left|\Sm^\dlj \right|)\gamamlpj}$
and $\mathbb{E}\big\{\|\mathbf{\hat{G}}_m^\dlj \mathbf{e}_{u}^\dlj\|^2\big\}= N \gamamlj$.

The received signal at the $k$-th downlink user, which is served by two distinct sets of APs, is given by 
\begin{align}\label{ykd}
y_k^\dld=&\sqrt{\eta\rho_{d}}\Big(\sum\nolimits_{m \in \mathcal{Z}_k^\dld}
 a_m (\gmkd)^H\bmkzf \nonumber\\
 &+ \sum\nolimits_{m' \in \bar{\mathcal{Z}}_k^\dld} a_{m'} (\gmpkd)^H\bmpkmr \Big)s_{k}^{\dld} \nonumber\\
 &+ \sum\nolimits_{k'\in\mathcal{K}, k' \neq k} \sqrt{\eta\rho_{d}} \Big(\sum\nolimits_{m \in \mathcal{Z}_{k'}^\dld}
a_m (\gmkd)^H\bmkpzf \nonumber\\
&+\sum\nolimits_{m \in\bar{\mathcal{Z}}_{k'}^\dld}a_m (\gmkd)^H \bmkpmr \Big)s_{k'}^{\dld}  \nonumber\\
&+  \sum\nolimits_{u \in\mathcal{U}} \sqrt{\eta\rho_{d}} \Big( \sum\nolimits_{m \in \mathcal{Z}_{u}^\dlj}
a_m 
(\gmkd)^H \bmlzf \nonumber\\
&+ \sum\nolimits_{m \in \bar{\mathcal{Z}}_{u}^\dlj}
a_m 
(\gmkd)^H \bmlmr
 \Big) s_{u}^{\dlj} \nonumber\\
 &+\sum\nolimits_{u \in\mathcal{U}} \sqrt{\rho_u}\guk s_{u}^{\lu}+w_{k}^\dld,
\end{align}
where $\mathcal{Z}_k^\dld$ denotes the set of APs employing PZF precoding for user $k$, and $\bar{\mathcal{Z}}_k^\dld$ denotes the set of APs applying MR precoding for the same user.
Here, $w_{k}^\dld \sim \mathcal{CN} (0,1)$ denotes the additive white Gaussian noise (AWGN) at the $k$-th downlink user and $s_{u}^{\lu}$ represents the symbol intended for transmission between a pair of untrusted users—from the $u$-th untrusted transmitter to the $u$-th untrusted receiver—with $\mathbb{E}\big\{|s_{u}^{\lu}|^2\big\}=1$ and $\mathbb{E}\left\{s_{u}^{\lu}\right\}=0$. The received signal at the $u$-th untrusted receiver, which is jammed by two distinct sets of APs, is given by 
\begin{align} \label{ylu}
y_u^{\lu}
= &\sqrt{\rho_u}\gulu s_{u}^{\lu}+ \sum\nolimits_{u'\in \mathcal{U}, u' \neq u} \sqrt{\rho_{u'}} \guplu s_{u'}^{\lu} \nonumber\\
&+ \sum\nolimits_{u' \in \mathcal{U}}  \sqrt{\eta\rho_{d}}\bigg(
\sum\nolimits_{m \in \mathcal{Z}_{u'}^\dlj} a_m 
(\gmlj)^H\bmlpzf s_{u'}^{\dlj} \nonumber\\
&+\sum\nolimits_{m' \in \bar{\mathcal{Z}}_{u'}^\dlj}a_{m'} (\gmplj)^H \bmplpmr s_{u'}^{\dlj} \bigg) \nonumber\\
&+\sum\nolimits_{k \in \mathcal{K}}  \sqrt{\eta\rho_{d}} \bigg(
\sum\nolimits_{m \in \mathcal{Z}_{k}^\dld}a_m  (\gmlj)^H \bmkzf s_{k}^{\dld} \nonumber\\
&+ \sum\nolimits_{m' \in \bar{\mathcal{Z}}_{k}^\dld} a_{m'}  (\gmplj)^H \bmpkmr s_{k}^{\dld}\bigg) + w_{u}^\lu,
\end{align}
where $\mathcal{Z}_{u}^\dlj$  and $\bar{\mathcal{Z}}_{u}^\dlj$ denote the set of APs employing PZF   precoding for jamming the $u$-th untrusted receiver, and    the set of APs applying MR  precoding, respectively,  while $w_{u}^\lu \sim \mathcal{CN} (0,1)$ denotes the AWGN  at the $u$-th untrusted receiver. 

 The signal received at the $m$-th AP in the monitoring mode to observe the untrusted receiver is given by
\begin{align}
 \qy_{m}^{\ul}
    &=\! \left(1\!-a_m\right)\sqrt{\rho_u} \gmuo s_{u}^{\lu} \!+\! \sum\nolimits_{u'\in \mathcal{U}}\!\left(1\!-\!a_m\right)\sqrt{\rho_{u'}} \gmupu s_{u'}^{\lu}\nonumber\\
    &+\sqrt{\eta\rho_d} \sum\nolimits_{i\in\mathcal{M}}a_i\left(1-a_m\right)\bigg(\sum\nolimits_{k'\in \Si^\dld}  \qF_{mi} \bikpzf s_{k'}^{\dld} \nonumber\\
    &+\sum\nolimits_{k\in \Wi^\dld} \qF_{mi} \bikmr s_{k}^{\dld} + \sum\nolimits_{u'\in\Si^\dlj} \qF_{mi} \bilpzf s_{u'}^{\dlj} \nonumber\\
&+\sum\nolimits_{u\in\Wi^\dlj} \qF_{mi} \bilmr s_{u}^{\dlj}\bigg) +\left(1-a_m\right)\qw_{m}^{\ul},
\end{align}
where $\qw_{m}^{\ul} \sim \mathcal{CN} (\mathbf{0},\mathbf{I}_{N})$ is the AWGN vector. 

In this paper, we consider  PZF   combining scheme at the APs in monitoring mode. In this case, the $m$-th AP with $a_m=0$ categorizes the untrusted transmitters into two groups based on their large-scale fading coefficients: 1) $\Sm^\ul$ comprising strong untrusted transmitters, and 2) $\Wm^\ul$ comprising weak untrusted transmitters. Then, it applies an equalizing linear combination to the received signal $\qy_{m}^{\ul}$ using a combining vector $\vmuo$, where 
\begin{align}
\vmuo \triangleq
\begin{cases}
  \vmumr, & \mbox{if}~u \in\Wm^\ul ,\\
  \vmuzf, & \mbox{if}~u\in\Sm^\ul.
\end{cases} 
 \end{align}
The combining vectors $\vmumr$ and $\vmuzf$, corresponding to the MR and ZF combining schemes, are given respectively by
$     \vmumr= \frac{\mathbf{\hat{G}}_m^\ul \mathbf{e}_{u}^\ul}{\sqrt{\mathbb{E}\big\{\left\|\mathbf{\hat{G}}_m^\ul \mathbf{e}_{u}^\ul\right\|^2\big\}}},$
and 
$     \vmuzf = \frac{\boldsymbol{\theta}_{mu}^\ul}{\sqrt{\mathbb{E} \big\{\|\boldsymbol{\theta}_{mu}^\ul \|^2\big\}}},$
where $\boldsymbol{\theta}_{mu}^\ul=\mathbf{\hat{G}}_m^\ul  \mathbf{\Upsilon}_{\Sm^\ul}^\ul ((\mathbf{\Upsilon}_{\Sm^\ul}^\ul)^H (\mathbf{\hat{G}}_m^\ul)^H\mathbf{\hat{G}}_m^\ul \mathbf{\Upsilon}_{\Sm^\ul}^\ul)^{-1} \boldsymbol{\varepsilon}_{u}^\ul$ and $\mathbf{\hat{G}}_m^\ul = \big[ \hgmoo,\dots, \hgmUo \big]_{(N \times U)}$. The vector $\mathbf{e}_{u}^\ul$ denotes the $u$-th column of the identity matrix $\mathbf{I}_{U}=\left[\mathbf{e}_{1}^\ul, \mathbf{e}_{2}^\ul, \dots, \mathbf{e}_{U}^\ul \right]_{(U \times U)}$, whereas the
matrix $\mathbf{\Upsilon}_{\Sm}^\ul=\big[ \mathbf{e}_{1'}^\ul, \mathbf{e}_{2'}^\ul, \dots, \mathbf{e}_{|\Sm^\ul|}^\ul \big]_{\left(U \times |\Sm^\ul| \right)}$ is constructed by selecting the columns of  $\mathbf{I}_{U}$ corresponding to the untrusted transmitter units in $\Sm^\ul$. The vector $\boldsymbol{\varepsilon}_{u'}^\ul$ denotes the $u'$-th column of the $\mathbf{\Upsilon}_{\Sm^\ul}^\ul$. In this case, we have
\begin{align}\label{smulu}
&\hat{s}_{mu}^{\lu}= (\vmuo)^H \qy_{m}^{\ul}.
\end{align}
 The resultant signal in \eqref{smulu} obtained at each $m$-th AP, operating in the monitoring mode,  is then forwarded to the CPU for detecting the untrusted transmitted symbol $s_{u}^{\lu}$.
At the CPU, a receiver combiner aggregates the  signals from all monitoring-mode APs. The final combined signal at the CPU is given by
 \begin{align}\label{sulu}
&\hat{s}_{u}^{\lu} =   \sum\nolimits_{m\in \mathcal{M}}  \hat{s}_{mu}^{\lu}.
\end{align}
\vspace{-1.7em}
\section{Performance Analysis and AP Mode Selection}\label{SINR}
In this section, we derive the effective SINR achieved at the $k$-th downlink communication user, and  the effective SINR for detecting the signal transmitted by the $u$-th untrusted transmitter at the CPU. We also derive the effective SINR experienced by the $u$-th untrusted receiver. By employing the widely used use-and-then-forget bounding technique~\cite{Du:TCOM:2021}, the SINR expressions for the $k$-th downlink user, the $u$-th untrusted receiver, and the CPU’s detection of the $u$-th untrusted transmitter are given in~\eqref{sinrkd}, \eqref{sinrlu}, and \eqref{sinruo}, on the top of the next page, respectively. Here, ${\text{DS}}_{k}^\dld$, ${\text{DS}}_{u}^\lu$, and ${\text{DS}}_{u}^\ul$ denote the desired signals for the $k$-th downlink user, the $u$-th untrusted receiver, and the $u$-th untrusted transmitter, respectively. The terms $\text{BU}_k^\dld$ and $\text{BU}_u^\ul$ represent the beamforming uncertainty for the downlink user and the monitored untrusted transmitter, respectively. The interference from other downlink users is captured by $\text{DI}_{k'}^\dld$, $\text{DI}_{k}^\lu$, and $\text{DI}_{k}^\ul$, while $\text{JI}_u^\dld$, $\text{JI}_{u'}^\lu$, and $\text{JI}_u^\ul$ denote the interference from jamming signals. Lastly, $\text{UI}_u^\dld$, $\text{UI}_{u'}^\lu$, and $\text{UI}_{u'}^\ul$ correspond to the interference from untrusted transmitters in the respective cases. The precise expressions for these terms are provided below.

\vspace{2em}
\begin{figure*}
\begin{align}
\mathrm{SINR}_k^\dld &= \frac{\left| {\text{DS}}_k^\dld \right|^2}{\mathbb{E} \left\{\left| \text{BU}_k^\dld \right|^2\right\} +\sum_{k' \in \mathcal {K}, k' \neq k} \mathbb{E} \left\{\left| \text{DI}_{k'}^\dld \right|^2\right\} + \sum_{u \in \mathcal{U}} \mathbb{E} \left\{\left| \text{JI}_u^\dld \right|^2\right\}+\sum_{u \in \mathcal{U}} \mathbb{E}\left\{\left|\text{UI}_u^\dld \right|^2\right\}+1},\label{sinrkd} ~\tag{9} \\[0.2em]
\mathrm{SINR}_{u}^\lu &= \frac{\left| {\text{DS}}_{u}^\lu \right|^2}{ \sum_{u'\in \mathcal{U}, u' \neq u} \mathbb{E}\left\{\left|\text{UI}_{u'}^\lu\right|^2\right\}+\sum_{u' \in \mathcal{U}} \mathbb{E}\left\{\left| \text{JI}_{u'}^\lu\right|^2\right\} +\sum_{k \in \mathcal {K}} \mathbb{E}\left\{\left| \text{DI}_{k}^\lu\right|^2\right\} +1},\label{sinrlu}  ~\tag{10}\\[0.2em]
\mathrm{SINR}_u^\ul &= \frac{\left| {\text{DS}}_u^\ul \right|^2}{\mathbb{E} \left\{\left| \text{BU}_u^\ul \right|^2\right\} +\sum_{u' \in \mathcal{U}, u' \neq u} \mathbb{E}\left\{\left|\text{UI}_{u'}^\ul\right|^2\right\}+\sum_{k \in \mathcal {K}} \mathbb{E} \left\{\left| \text{DI}_{k}^\ul \right|^2\right\} + \sum_{u\in \mathcal{U}} \mathbb{E} \left\{\left| \text{JI}_u^\ul \right|^2\right\}+\mathbb{E} \left\{\left|\text{N}_u\right|^2\right\}},\label{sinruo} ~\tag{11}
\end{align} 
\hrulefill
\vspace{-1.5em}
\end{figure*}
\vspace{-2em}
\subsection{Received SINR at the $k$-th Downlink Communication User}\label{sinrmr}
 Using~\eqref{ykd}, the corresponding SINR terms can be written as:
\setcounter{equation}{11}
\begin{subequations}
	\allowdisplaybreaks
\begin{align}
&{\text{DS}}_k^\dld= \sqrt{\eta\rho_{d}}~\mathbb{E} \Big\{\sum\nolimits_{m \in \mathcal{Z}_k^\dld}
 a_m (\gmkd)^H\bmkzf\nonumber\\
 &\hspace{2em}+ \sum\nolimits_{m' \in \bar{\mathcal{Z}}_k^\dld} a_{m'} (\gmpkd)^H\bmpkmr \Big\}, 
\\
&\text{BU}_k^\dld = \sqrt{\eta\rho_{d}}\sum\nolimits_{m \in \mathcal{Z}_k^\dld}
 a_m (\gmkd)^H\bmkzf\nonumber\\
 &\hspace{1em}+\sqrt{\eta\rho_{d}} \sum\nolimits_{m' \in \bar{\mathcal{Z}}_k^\dld} a_{m'} (\gmpkd)^H\bmpkmr -{\text{DS}}_k^\dld,
\\
&\text{DI}_{k'}^\dld = \sqrt{\eta\rho_{d}} \sum\nolimits_{m \in \mathcal{Z}_{k'}^\dld} a_m (\gmkd)^H \bmkpzf \nonumber\\
& \hspace{2em}+ \sqrt{\eta\rho_{d}} \sum\nolimits_{m' \in \bar{\mathcal{Z}}_{k'}^\dld} a_{m'} (\gmpkd)^H \bmpkpmr,
\\
&\text{JI}_u^\dld =  \sqrt{\eta\rho_{d}} \sum\nolimits_{m \in \mathcal{Z}_{u}^\dlj} a_m (\gmkd)^H \bmlzf \nonumber\\
&\hspace{2em}+ \sqrt{\eta\rho_{d}} \sum\nolimits_{m' \in \bar{\mathcal{Z}}_{u}^\dlj} a_{m'} (\gmpkd)^H \bmplmr,\\
&\text{UI}_u^\dld = \sqrt{\rho_u}\guk.
\end{align}   
\end{subequations}
By calculating the corresponding expected values in~\eqref{sinrkd},
the SINR at the $k$-th downlink communication user  can be obtained as in the following
proposition.
\vspace{-0.5em}
\begin{proposition}
The SE at the $k$-th downlink communication users is   $\mathrm{SE}_k^\dld  = \frac{T-\tau}{T} \log_{2}\left({1+\mathrm{SINR}_k^\dld} \right)$, where the closed-form expressions for the effective SINR at the $k$-th downlink user, $\mathrm{SINR}_k^\dld$, is given by \eqref{sinrkda}  at the top of the next page.    
\end{proposition}
\begin{proof}
 The proof is omitted due to page constraints. 
\end{proof}
\vspace{-1em}
\subsection{Received SINR at the $u$-th Untrusted Receiver }\label{sinrzf}
From \eqref{ylu}, the corresponding SINR terms for the $u$-th untrusted receiver can be written as follows:
\setcounter{equation}{13}
\vspace{-0.1em}
\begin{subequations}\label{eq:dsmkv}
	\allowdisplaybreaks
\begin{align}
& {\text{DS}}_{u}^\lu = \sqrt{\rho_u}    \gulu, \\
&\text{UI}_{u'}^\lu =  \sqrt{\rho_{u'}} \guplu, \\
&\text{JI}_{u'}^\lu =  \sqrt{\eta\rho_{d}}\Big(
\sum\nolimits_{m \in \mathcal{Z}_{u'}^\dlj} a_m 
(\gmlj)^H\bmlpzf   \nonumber\\
&\hspace{5em}+\sum\nolimits_{m' \in \bar{\mathcal{Z}}_{u'}^\dlj}a_{m'} (\gmplj)^H \bmplpmr   \Big),\\
&\text{DI}_{k}^\lu=\sqrt{\eta\rho_{d}}\sum\nolimits_{m \in \mathcal{Z}_{k}^\dld}a_m  (\gmlj)^H \bmkzf  \nonumber\\
&\hspace{3em}+ \sqrt{\eta\rho_{d}}\sum\nolimits_{m' \in \bar{\mathcal{Z}}_{k}^\dld} a_{m'}  (\gmplj)^H \bmpkmr.
\end{align}    
\end{subequations}
We assume that the untrusted receivers have perfect CSI which represents the worst-case scenario from a monitoring performance perspective \cite{Zahra:IOT:2024}.
\vspace{-0.2em}
\begin{proposition}
The received SINR for the  untrusted link
at the u-th untrusted receiver is given by
    \eqref{sinrlua}   at the top of the next page.    
\end{proposition}
\begin{proof}
 The proof is omitted due to page constraints. 
\end{proof}
\vspace{-1em}
\subsection{Received SINR at the CPU for Observing the $u$-th Untrusted Transmitter}
From \eqref{sulu}, the corresponding SINR terms at the CPU for observing the $u$-th  untrusted transmitter can be written as follows:
\setcounter{equation}{15}
\vspace{-0.1em}
\begin{subequations}
	\allowdisplaybreaks
\begin{align} 
&{\text{DS}}_u^\ul= \sqrt{\rho_u}~\mathbb{E} \Big\{\sum\nolimits_{m\in \mathcal{Z}_u^\ul}  (1-a_m) (\vmuzf)^H \gmuo \nonumber\\
&\hspace{2em}+\sum\nolimits_{m'\in \bar{\mathcal{Z}}_u^\ul}(1-a_{m'}) (\vmpumr)^H \gmpudu \Big\} , 
\\
&\text{BU}_u^\ul = \sqrt{\rho_u}\sum\nolimits_{m\in \mathcal{Z}_u^\ul}  (1-a_m) (\vmuzf)^H \gmuo s_{u}^{\lu} \nonumber\\
 &\hspace{1.5em}+ \sqrt{\rho_u}\sum\nolimits_{m'\in \bar{\mathcal{Z}}_u^\ul}  (1-a_{m'}) (\vmpumr)^H \gmpudu s_{u}^{\lu} \!- {\text{DS}}_u^\ul,\\
&\text{UI}_{u'}^\ul = \sqrt{\rho_{u'}} \sum\nolimits_{m\in \mathcal{Z}_u^\ul} (1-a_m) (\vmuzf)^H \gmupu s_{u'}^{\lu}\nonumber\\
&\hspace{3em} + \sqrt{\rho_{u'}}\sum\nolimits_{m'\in \bar{\mathcal{Z}}_u^\ul} (1-a_{m'}) (\vmpumr)^H \gmpupu s_{u'}^{\lu}, 
\\
&\text{DI}_{k}^\ul = \sqrt{\eta\rho_d} \sum\nolimits_{m\in \mathcal{Z}_u^\ul} (1-a_m) \Big(\sum_{i\in\mathcal{M}} a_i  (\vmuzf)^H \qF_{mi} \bikzf \nonumber\\
& \hspace{0.5em}
+ \sum\nolimits_{i'\in\mathcal{M}} a_{i'}  (\vmuzf)^H\qF_{mi'} \bipkmr \Big) 
\nonumber\\
& \hspace{0.5em}+\!\! \sqrt{\eta\rho_d} \sum\nolimits_{m'\in \bar{\mathcal{Z}}_u^\ul}\!(1\!-\!a_{m'})\Big( \!\sum\nolimits_{i\in\mathcal{M}}\! \!a_i  (\vmpumr)^H\qF_{m'i} \bikzf \nonumber\\
&  \hspace{0.5em} \!+ \!\sum\nolimits_{i'\in\mathcal{M}} a_{i'}  (\vmpumr)^H\qF_{m'i'} \bipkmr \Big),
\\
&\text{JI}_u^\ul \!=  \!\sqrt{\eta\rho_{d}} \sum\nolimits_{m\in \mathcal{Z}_u^\ul}\! (1\!-\!a_m) \Big( \sum\nolimits_{i\in\mathcal{M}} \!\!a_i  (\vmuzf)^H\qF_{mi} \bilzf 
\nonumber\\
&\hspace{0.5em}+ \sum\nolimits_{i'\in\mathcal{M}} a_{i'} (\vmuzf)^H\qF_{mi'} \biplmr\Big) 
\nonumber\\
&\hspace{0.5em}+ \sqrt{\eta\rho_{d}} \sum_{m'\in \bar{\mathcal{Z}}_u^\ul} (1-a_{m'}) \Big(\sum\nolimits_{i\in\mathcal{M}} a_i  (\vmpumr)^H\qF_{m'i} \bilzf \nonumber\\
&\hspace{0.5em}+ \sum\nolimits_{i'\in\mathcal{M}} a_{i'} (\vmpumr)^H \qF_{m'i'} \biplmr \Big), \\
& \text{N}_u= \sum\nolimits_{m\in \mathcal{Z}_u^\ul} (1-a_m) (\vmuzf)^H\qw_{m}^{\ul} \nonumber\\
&\hspace{2em}+\sum\nolimits_{m'\in \bar{\mathcal{Z}}_u^\ul} (1-a_{m'}) (\vmpumr)^H \qw_{m'}^{\ul}.
\end{align}    
\end{subequations}
\begin{proposition}
The closed-form expressions for the effective SINR   at CPU for observing the $u$-th untrusted transmitter, $\mathrm{SINR}_u^\ul$, is given by \eqref{sinruoa} at the top of the page.    
\end{proposition}
\begin{proof}
 The proof is omitted due to page constraints. 
\end{proof}

\vspace{-0.8em}
\subsection{Monitoring Success Probability}
The reliability of the MSP at the CPU relies on the SINR achieved at the untrusted receiver, $\mathrm{SINR}_{u}^\UR$, and the SINR achieved by the CPU's observation, $\mathrm{SINR}_u^\ul$. The condition for successful monitoring at the CPU is defined as 
\setcounter{equation}{17}
\begin{align}
    \Omega_u = \begin{cases}
        1, &  \mathrm{SINR}_u^\ul \geq \mathrm{SINR}_{u}^\UR, \\
        0, & \mbox{otherwise}.
    \end{cases}
\end{align}
Here, $\Omega_u=1$ indicates a successful monitoring event of the $u$-th untrusted transmitter, while $\Omega_u=0$ represents a monitoring failure. Accordingly, the MSP is defined as the expectation of $\Omega_u$ and is given by
\vspace{-0.1em}
\begin{align}
    \text{MSP}_u=\mathbb{E} \left\{\Omega_u \right\} &=  \Pr\left( \mathrm{SINR}_u^\ul \geq \mathrm{SINR}_{u}^\lu \right).
\end{align}
Let the denominator of \eqref{sinrlua} be denoted by $\Gamma_u$. Then, the closed-form expression for the MSP is given by
\vspace{-0.1em}
\begin{align}\label{MSP0}
    \Pr\Big( \mathrm{SINR}_u^\ul \geq \frac{\rho_u | \gulu|^2}{\Gamma_u} \Big) = 1 - \exp\Big( - \frac{\mathrm{SINR}_u^\ul \Gamma_u}{\rho_u \betaulu}\Big).
\end{align}
\begin{figure*}
\begin{align} 
\mathrm{SINR}_k^\dld &= \frac{\eta\rho_{d}\Big( \sum\nolimits_{m \in \mathcal{Z}_k^\dld}
 a_m \sqrt{\Big(N-\big|\Sm^\dld\big|\Big) \gamamkd} + \sum\nolimits_{m' \in \bar{\mathcal{Z}}_k^\dld} a_{m'} \sqrt{N  \gamampkd} \Big)^2}{\eta\rho_{d}\Big( \sum\nolimits_{k' \in \mathcal {K}} \sum\nolimits_{m \in \mathcal{M}} a_{m} \betamkd \!-\! \sum\nolimits_{k' \in \mathcal {K}} \sum\nolimits_{m \in \mathcal{Z}_{k'}^\dld}
 a_m \gamamkd+ \sum\nolimits_{u \in \mathcal{U}} \sum\nolimits_{m \in \mathcal{M}} a_m\  \betamkd \Big) \!+\!\sum\nolimits_{u \in \mathcal{U}} \rho_u \beta_{uk}\!+\!1},\label{sinrkda} ~\tag{13}\\[-0.2em]
\mathrm{SINR}_{u}^\UR &= \frac{\rho_u | \gulu|^2}{\sum\nolimits_{u'\in \mathcal{U}, u' \neq u} \rho_{u'}\betauplu \!+\!\eta\rho_{d} \Big(\! \sum\nolimits_{u' \in \mathcal{U}} \sum\nolimits_{m \in \mathcal{M}} a_{m}  \betamlj \!-\!\sum\nolimits_{u' \in \mathcal{U}} 
\sum\nolimits_{m \in \mathcal{Z}_{u'}^\dlj} a_m  \gamamlj  +\sum\nolimits_{k \in \mathcal {K}} \sum\nolimits_{m \in \mathcal{M}}a_m  \betamlj \Big) +1}, \label{sinrlua} ~\tag{15}\\[-0.2em]
\mathrm{SINR}_u^\ul&= \frac{\rho_u\Big( \sum\nolimits_{m\in \mathcal{Z}_u^\ul} (1\!-\!a_m) \sqrt{\Big(N\!-\!\big|\Sm^\ul\big|\Big) \gamamuo}+ \sum\nolimits_{m'\in \bar{\mathcal{Z}}_u^\ul}  (1\!-\!a_{m'}) \sqrt{N \gamampuo} \Big)^2}{\! \sum\limits_{u' \in \mathcal{U}}\!\!\rho_{u'}\!\!\sum\limits_{m\in \mathcal{M}} \!\!(1\!-a_m) \betamupo \!- \!\sum\limits_{u' \in \mathcal{U}}\! \!\rho_{u'}\!\!\sum\nolimits_{m'\in \bar{\mathcal{Z}}_u^\ul}  (1\!-\!a_{m'}\!) \gamampupo\!+ \eta\rho_d \!\!\sum\limits_{m \in \mathcal{M}}\!\sum\limits_{i\in\mathcal{M}} \! (1\!-\!a_m)  a_i  \betami \!+\!\!\sum\limits_{m\in \mathcal{M}} \!\!(1\!-a_m)}, \label{sinruoa} ~\tag{17}
\end{align}
\vspace{-1.8em}
\hrulefill
\end{figure*}
\vspace{-1.4em}
\subsection{AP Mode Assignment}
Our objective is to enhance the minimum MSP by appropriately assigning APs to either monitoring or downlink modes, represented by the binary variables $\boldsymbol{a} \triangleq \{a_m\}$. This assignment is subject to the minimum QoS requirement, $\mathrm{SE}_{QoS}^\dld$, for each downlink user. To this end, we propose a simple yet effective algorithm for AP mode assignment using the derived MSP in~\eqref{MSP0}. \textbf{Algorithm~1} presents a greedy approach for AP mode selection. Let $\mathcal{M}_{\text{mo}}$ and $\mathcal{M}_{\text{dl}}$   denote the sets containing the indices of APs in monitoring    mode with $a_m=0$, and the indices of APs in downlink mode with $a_m=1$, respectively. Initially, all APs are assigned to the downlink mode, i.e., $a_m=1, \forall m$,  and hence $\mathcal{M}_{\text{dl}} = \mathcal{M}$ and $\mathcal{M}_{\text{mo}} = \emptyset$. In each iteration, the algorithm selects one AP from the downlink set that yields the largest monitoring gain while still satisfying the minimum SE requirements for all downlink users, and reassigns it to the monitoring set. In step 6, $\mathrm{SE}_k^\dld$  is calculated  using~\eqref{sinrkda}. The computational complexity of \textbf{Algorithm~1} is $\mathcal{O}(UM^2)$.

 \begin{algorithm}[!t]
\caption{Greedy AP Mode Assignment}
\begin{algorithmic}[1]
\label{alg:Grreedy} 
\STATE
\textbf{Initialize}: Set  $\mathcal{M}_{\text{dl}}=\mathcal{M}$ and $\mathcal{M}_{\text{mo}}=\emptyset$. Set iteration index $i=0$.
 Calculate $\Pi^\star[i]=  		\underset{u\in\mathcal{U}} \min \,\,\mathbb{E} \left\{\Omega_u  (\mathcal{M}_{\text{mo}}, \mathcal{M}_{\text{dl}})\right\}$
\REPEAT
\FORALL{$m \in \mathcal{M}_{\text{dl}}$}
\STATE Set $\mathcal{M}_{\text{mo}}=\mathcal{M}_{\text{mo}} \bigcup m$ and $\mathcal{M}_{\text{dl}}=\mathcal{M}_{\text{dl}} \setminus m$.
\STATE  Calculate $\Pi_m=  		\underset{u\in\mathcal{U}} \min \,\,\mathbb{E} \left\{\Omega_u  (\mathcal{M}_{\text{dl}} \setminus m, \mathcal{M}_{\text{mo}}\bigcup m)\right\}$
\STATE Calculate $\Xi_m = \underset{k\in\mathcal{K}}\min\,\, \mathrm{SE}_k^\dld(\mathcal{M}_{\text{dl}} \setminus m, \mathcal{M}_{\text{mo}}\bigcup m)$\\
\ENDFOR
\STATE Set $\Pi^\star[i+1]= \underset{m\in\mathcal{M}_{\text{mo}}} \max \,\,\Pi_m$,~~ $e=|\Pi^\star[i+1]- \Pi^\star[i]|$ 
\IF{$e \geq e_{\min} ~\AND ~\Xi_m \geq \mathrm{SE}_{QoS}^\dld$} 
\STATE Select AP $m^\star=\argmax_{m\in\mathcal{M}_{\text{mo}}}\{\Pi_m\}$
\STATE {Update $\mathcal{M}_{\text{mo}}=\{\mathcal{M}_{\text{mo}}\bigcup m^{\star}\}$ and $\mathcal{M}_{\text{dl}}=\mathcal{M}_{\text{dl}}\setminus m^{\star}$}
\ENDIF
\UNTIL{ $e < e_{\min}$ }
\RETURN $\mathcal{M}_{\text{mo}}$ and $\mathcal{M}_{\text{dl}}$, i.e., the indices of APs in monitoring mode and downlink mode, respectively.
\end{algorithmic}
\end{algorithm}
\vspace{-0.5em}
\section{Numerical Results}
We assume that $M$ APs, $K$  legitimate users,  $U$ untrusted receivers, and $U$ untrusted transmitters are  randomly distributed within an area of size $1 \times 1$ km$^2$. Also, $N=6$, the maximum transmission power for each AP is $1$ W, and for each untrusted transmitter is $0.2$ W, the noise power is $-92$ dBm, while $B=50$ MHz, unless otherwise stated. Moreover, $\beta_{m,i}$ is modeled following \cite{Björnson:TWC:2020}, i.e.,
$\beta_{m,i} = 10^{\frac{\text{PL}_{m,i}^d}{10}}10^{\frac{F_{m,i}}{10}}$, where $10^{\frac{\text{PL}_{m,i}^d}{10}}$ is the path loss, $10^{\frac{F_{m,i}}{10}}$ denotes the shadowing effect with $F_{m,i}\in\mathcal{N}(0,4^2)$ (in dB). Also, $\text{PL}_{m,i}^d$ is in dB and can be calculated as
$\text{PL}_{m,i}^d = -30.5-36.7\log_{10}(d_{m,i}/{1\,\text{m}})$.
The correlation among the shadowing terms from the $m$-th AP to   $g \in \{\mathcal{K} \cup \mathcal{U}\}$  downlink users, untrusted receivers, and untrusted transmitters can be given by  $\mathbb{E}\{F_{m,g}F_{j,g'}\}=4^22^{-\upsilon_{g,g'}/9\text{m}}$, if $j = m$, and $\mathbb{E}\{F_{m,g}F_{j,g'}\}=0$, if $j \neq m$, where $\upsilon_{g,g'}$ is the physical distance between users $g$ and $g'$. Additionally, for a fair comparison, the minimum QoS requirement, $\mathrm{SE}_{QoS}^{\mathrm{\dld}}$, is set equal to the minimum SE obtained through random AP mode selection.

Figure~\ref{fig2} illustrates the performance of the CF-mMIMO JCAM system employing \textbf{Algorithm 1} for AP mode assignment under varying numbers of downlink legitimate users $K$ and untrusted links $U$. In this figure, we also compare the MSP of the JCAM scheme against that of a co-located FD massive MIMO system, where all APs are co-located as an antenna array that simultaneously performs proactive monitoring and communication at the same frequency. For fair comparison, the co-located system deploys $\frac{MN}{2}$ antennas for observing, while the remaining $\frac{MN}{2}$ antennas are used for downlink communication and jamming.
The performance of JCAM with random AP mode selection, where the AP mode is chosen randomly, is also included in the figure. The results demonstrate that our proposed JCAM framework achieves nearly a six-fold improvement in the minimum MSP compared to the co-located massive MIMO baseline. Moreover, the average minimum MSP obtained with \textbf{Algorithm 1} consistently outperforms the baseline random AP mode assignment strategy. Notably, when the system has a small number of APs, \textbf{Algorithm 1} provides up to a $32\%$ improvement in the minimum MSP. Even as the number of APs increases, the proposed approach continues to yield superior performance relative to the baseline. In addition, the results confirm the advantages of deploying a large number of APs in CF-mMIMO systems, as the MSP increases significantly with larger $M$ (at the expense of increased complexity).

Figure~\ref{fig4} investigates the impact of the number of antennas per AP on the minimum MSP performance of the proposed CF-mMIMO-based JCAM system, evaluated using our proposed \textbf{Algorithm 1}. The total number of antennas in the system is fixed at $N_{total}=M\times N=120$, with different numbers of APs. The results reveal that as the number of antennas per AP increases, under a fixed total antenna budget, the performance of the random mode assignment degrades. This is primarily due to a reduction in the number of APs available for monitoring. In contrast, the performance degradation in the proposed \textbf{Algorithm 1} is significantly smaller. Notably, when the number of downlink users and untrusted pairs are reduced, our proposed algorithm not only avoids degradation but also achieves performance gains in mode assignment. This highlights the algorithm’s adaptability in effectively managing limited system resources.
\begin{figure}[t] 
\centering
\includegraphics[width=0.4\textwidth]{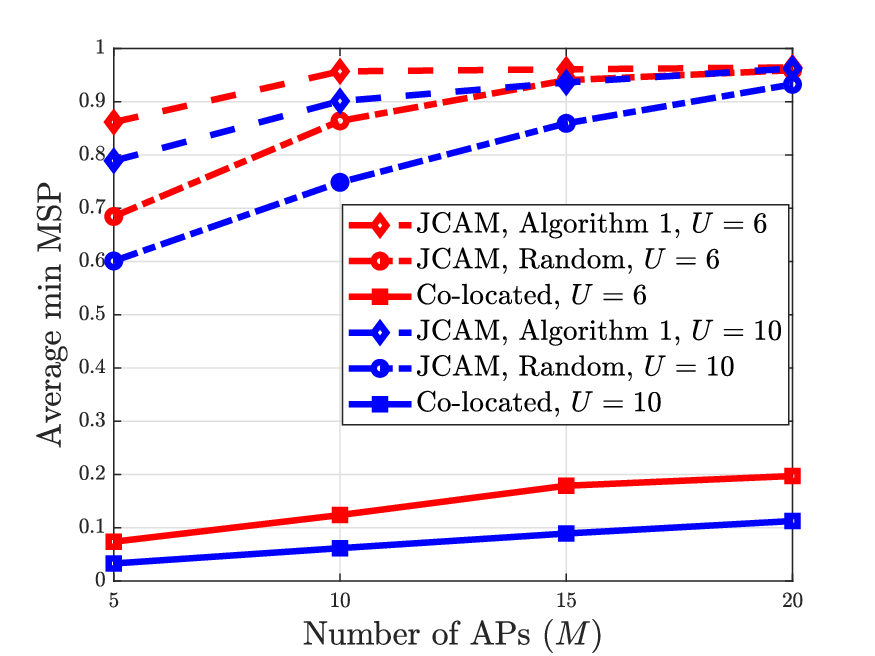}
\centering
\vspace{-0.5em}
\caption{Average minimum MSP against the number of APs for $N=6$, $U=K$.}
\label{fig2}
\vspace{-0.5em}
\end{figure}

\vspace{-1em}
\section{Conclusion}
In this paper, we proposed a novel JCAM system that integrates communication and proactive monitoring functionalities within a CF-mMIMO architecture. The proposed system enables APs to monitor multiple untrusted links while simultaneously providing communication services to multiple legitimate users. Our analytical framework, which includes closed-form expressions for SE and MSP, led to a simple yet efficient AP mode assignment method. Numerical results validated the benefits of the proposed CF-mMIMO-based JCAM framework, demonstrating significant improvements in monitoring performance over existing benchmarks, while satisfying the QoS requirements of each legitimate user. Therefore, the proposed JCAM system offers a promising solution for future wireless networks where both security and communication requirements are critical.

\begin{figure}[t] 
\centering
\includegraphics[width=0.4\textwidth]{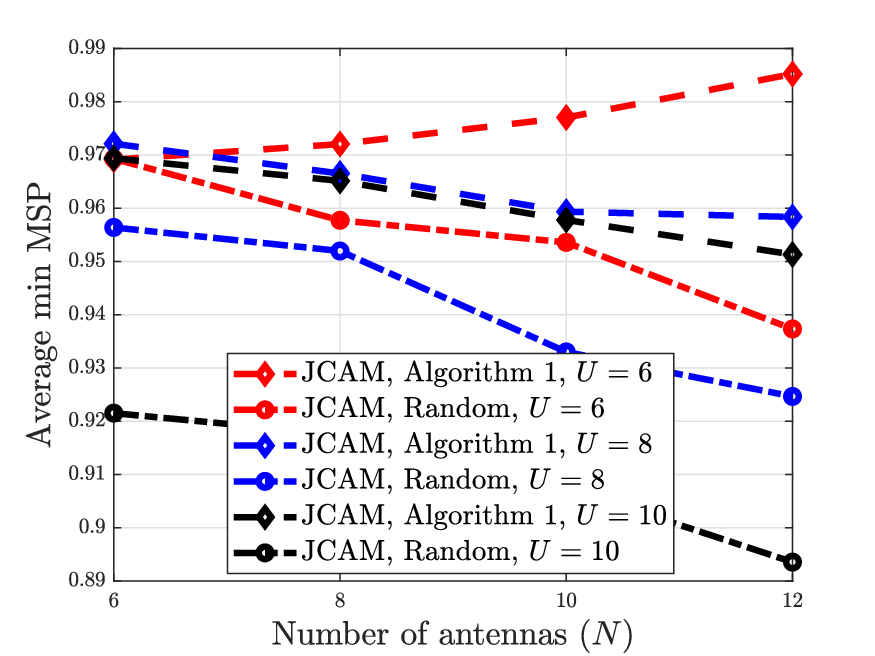}
\centering
\vspace{-0.5em}
\caption{Average minimum MSP against the number of antennas $N$ ($N_{total} = 120$, $U=K$).}
\vspace{-.5em}
\label{fig4}
\end{figure}


\vspace{-1em}
\bibliographystyle{IEEEtran}
\bibliography{bibliography}
\end{document}